\documentclass[USenglish,a4paper,cleveref]{lipics-v2019}
\usepackage{microtype}
\usepackage{xspace}

\definecolor{blueLink}{rgb}{0,0.2,0.8}
\hypersetup{colorlinks,linkcolor=blueLink,urlcolor=blueLink,citecolor=blueLink}

\newcommand{\E}{\ensuremath{\mathbf{E}}\xspace}
\newcommand{\LB}{\ensuremath{Q}\xspace}
\newcommand{\OPT}{\textsc{Opt}\xspace}
\newcommand{\ALG}{\textsc{Alg}\xspace}
\newcommand{\GKS}{{generalized $k$-server problem}\xspace}
\newcommand{\dist}{\mathsf{dist}}
\newcommand{\mindist}{\mathsf{mindist}}
\newcommand{\level}{\mathsf{level}}
\newcommand{\rank}{\mathsf{rank}}
\newcommand{\comp}{\ensuremath{\mathsf{comp}}}
\newcommand{\wc}{\ensuremath{\star}\xspace}
\newcommand{\SET}{\ensuremath{\mathcal{A}}\xspace}
\newcommand{\alghyd}{\ensuremath{\mathrm{H}}\xspace}
\newcommand{\alggks}{\ensuremath{\mathrm{G}}\xspace}
\newcommand{\con}{\textsc{Confine}\xspace}
\newcommand{\ouralg}{\ensuremath{\textsc{Herc}}\xspace}
\newcommand{\prob}{\eta}
\newcommand{\rprob}{\mu}
\newcommand{\tree}{\ensuremath{T}\xspace}
\newcommand{\treeroot}{\ensuremath{r_T}\xspace}
\newcommand{\height}{\ensuremath{h_T}\xspace}
\newcommand{\leaves}{\ensuremath{L_T}\xspace}
\newcommand{\DET}{\textsc{Det}\xspace}
\newcommand{\RAND}{\textsc{Rand}\xspace}
\newcommand{\e}{\ensuremath{\mathrm{e}}}
\newcommand{\eps}{\ensuremath{\varepsilon}}
\newcommand{\cspace}[1]{\ensuremath{S[#1]}\xspace}

\theoremstyle{plain}
\newtheorem{observation}[theorem]{Observation}

\title{Slaying Hydrae: Improved Bounds for Generalized k-Server in Uniform Metrics}

\author{Marcin Bienkowski}{Institute of Computer Science, University of Wrocław, Poland}{marcin.bienkowski@cs.uni.wroc.pl}{https://orcid.org/0000-0002-2453-7772}{}
\author{Łukasz Jeż}{Institute of Computer Science, University of Wrocław, Poland}{lukasz.jez@cs.uni.wroc.pl}{https://orcid.org/0000-0002-7375-0641}{}
\author{Paweł Schmidt}{Institute of Computer Science, University of Wrocław, Poland}{pawel.schmidt@cs.uni.wroc.pl}{}{}

\funding{Supported by Polish National Science Centre grant 2016/22/E/ST6/00499.}

\authorrunning{M. Bienkowski, Ł. Jeż and P. Schmidt}
\Copyright{Marcin Bienkowski, Łukasz Jeż, Paweł Schmidt}


\EventEditors{Pinyan Lu and Guochuan Zhang}
\EventNoEds{2}
\EventLongTitle{30th International Symposium on Algorithms and Computation (ISAAC 2019)}
\EventShortTitle{ISAAC 2019}
\EventAcronym{ISAAC}
\EventYear{2019}
\EventDate{December 8--11, 2019}
\EventLocation{Shanghai University of Finance and Economics, Shanghai, China}
\EventLogo{}
\SeriesVolume{149}
\ArticleNo{17}

\nolinenumbers

\begin{document}

\maketitle
\begin{abstract}
The \GKS is an extension of the weighted $k$-server problem, which in turn
extends the classic $k$-server problem. In the \GKS, each of $k$ servers $s_1,
\dots, s_k$ remains in its own metric space $M_i$.  A~request is a~tuple
$(r_1,\dots,r_k)$, where $r_i \in M_i$, and to service it, an algorithm needs to
move at least one server $s_i$ to the point $r_i$. The objective is to
minimize the total distance traveled by all servers.

In this paper, we focus on the \GKS for the case where all $M_i$ are uniform metrics.
We show an $O(k^2 \cdot \log k)$-competitive randomized
algorithm improving over a~recent result by Bansal et al. [SODA 2018], who
gave an $O(k^3 \cdot \log k)$-competitive algorithm. 
To this end, we define an abstract online problem, called Hydra game, and we
show that a~randomized solution of low cost to this game implies a 
randomized algorithm to the \GKS with low competitive ratio. 

We also show that no randomized algorithm can achieve competitive ratio lower than 
$\Omega(k)$, thus improving the lower bound of $\Omega(k / \log^2 k)$ by Bansal et al.
\end{abstract}

\section{Introduction}

The $k$-server problem, introduced by Manasse et al.~\cite{MaMcSl90}, is one
of the most well-studied and influential cornerstones of online analysis. The
problem definition is deceivingly simple: There are $k$ servers, starting at a
fixed set of $k$ points of a~metric space $M$. An input is a~sequence of
requests (points of $M$) and to service a request, an algorithm needs to move
servers, so that at least one server ends at the request position. As typical
for online problems, the $k$-server problem is sequential in nature: an online
algorithm \ALG learns a new request only after it services the current one. The
cost of \ALG, defined as the total distance traveled by all its servers, is then
compared to the cost of an \emph{offline} solution \OPT; the ratio between
them, called \emph{competitive ratio}, is subject to minimization.

In a natural extension of the $k$-server problem, called the
\emph{\GKS}~\cite{KouTay04,SitSto06}, each server $s_i$ remains in its own
metric space $M_i$. The request is a $k$-tuple $(r_1,\dots,r_k)$, where $r_i
\in M_i$, and to service it, an algorithm needs to move servers, so that 
\emph{at least one} server $s_i$ ends at the request position $r_i$. The
original $k$-server problem corresponds to the case where all metric
spaces~$M_i$ are identical and each request is of the form $(r,\dots,r)$. The
\GKS contains many known online problems, such as
the weighted $k$-server problem~\cite{BaElKo17,ChiVis13,ChrSga04,FiaRic94} or the
CNN problem~\cite{Chroba03,KouTay04,Sitter14,SitSto06} as special cases.

So far, the existence of an $f(k)$-competitive algorithm for the \GKS in
arbitrary metric spaces remains open. Furthermore, even for specific spaces,
such as the line~\cite{KouTay04} or uniform metrics~\cite{BaElKo17,BaElKN18,KouTay04},
the \GKS requires techniques substantially different from those used to tackle
the classic $k$-server problems. For these reasons, studying this problem
could lead to new techniques for designing online algorithms.

%%%%%%%%%%%%%%%%%%%%%%%%%%%%%%%%%%%%%%%%%%%%%%%%%%%%%%%%%%%%%%%%%%%%%%%%%%%%%%%

\subsection{Previous Work}

After almost three decades of extensive research counted in dozens of
publications (see, e.g., a~slightly dated survey by Koutsoupias
\cite{Koutso09}), we are closer to understanding the nature of the classic
$k$-server problem. The competitive ratio achievable by deterministic algorithms
is between $k$~\cite{MaMcSl90} and $2k-1$~\cite{KouPap95} with $k$-competitive
algorithms known for special cases, such as uniform metrics~\cite{SleTar85},
lines and trees~\cite{ChKaPV91,ChrLar91}, or metrics of $k+1$
points~\cite{MaMcSl90}. Less is known about competitive ratios for randomized
algorithms: the best known lower bound holding for an arbitrary metric space is
$\Omega(\log k / \log \log k)$~\cite{BaLiMN03} and the currently best upper
bound of~$O(\log^6 k)$ has been recently obtained in a breakthrough
result~\cite{BuCLLM18,Lee18}.

In comparison, little is known about the \GKS. In
particular, algorithms attaining competitive ratios that are functions of $k$
exist only in a few special cases. The case of $k = 2$ has been solved by
Sitters and Stougie~\cite{SitSto06,Sitter14}, who gave constant competitive
algorithms for this setting. Results for $k \geq 3$ are known only for simpler
metric spaces, as described below.

\begin{description}

\item[A \emph{uniform metric} case] describes a scenario where all metrics $M_i$
are uniform with pairwise distances between different points equal to $1$. For
this case, Bansal et al.~\cite{BaElKN18} recently presented an~$O(k \cdot
2^k)$-competitive deterministic algorithm and an $O(k^3 \cdot \log
k)$-competitive randomized one. The deterministic competitive ratio is at least
$2^k - 1$ already when metrics $M_i$ have two points~\cite{KouTay04}.
Furthermore, using a straightforward reduction to the metrical task system (MTS)
problem~\cite{BoLiSa92}, they show that the randomized competitive ratio is at
least $\Omega(k / \log k)$~\cite{BaElKN18}.\footnote{In fact, for the \GKS in
uniform metrics, the paper by Bansal et al.~\cite{BaElKN18} claims only the
randomized lower bound of $\Omega(k / \log^2 k)$. To obtain it, they reduce the
problem to the $n$-state metrical task system (MTS) problem and apply a lower
bound of $\Omega(\log n / (\log \log n)^2)$ for MTS~\cite{BaBoMe06}. By using
their reduction and a~stronger lower bound of $\Omega(\log n / \log \log n)$ for
$n$-state MTS~\cite{BaLiMN03}, one could immediately obtain a lower bound of
$\Omega(k / \log k)$ for the \GKS.} 

\item[A \emph{weighted uniform metric} case] describes a scenario where each
metric $M_i$ is uniform, but they have different scales, i.e., the pairwise
distances between points of $M_i$ are equal to some values~$w_i>0$. For this
setting, Bansal et al.~\cite{BaElKN18} gave an $2^{2^{O(k)}}$-competitive
deterministic algorithm extending an $2^{2^{O(k)}}$-competitive algorithm for
the weighted $k$-server problem in uniform metrics~\cite{FiaRic94}. (The latter
problem corresponds to the case where all requests are of the form
$(r,\dots,r)$.) This matches a lower bound of
$2^{2^{\Omega(k)}}$~\cite{BaElKo17} (which also holds already for the weighted
$k$-server problem).

\end{description}

%%%%%%%%%%%%%%%%%%%%%%%%%%%%%%%%%%%%%%%%%%%%%%%%%%%%%%%%%%%%%%%%%%%%%%

\subsection{Our Results and Paper Organization}

In this paper, we study the uniform metric case of the \GKS. We give a
randomized $O(k^2 \cdot \log k)$-competitive algorithm improving over the
$O(k^3 \cdot \log k)$ bound by Bansal et al.~\cite{BaElKN18}.

To this end, we first define an elegant abstract online problem: a~\emph{Hydra
game} played by an~online algorithm against an adversary on an unweighted
tree. We present the problem along with a randomized, low-cost online
algorithm \ouralg in \cref{sec:hydra}. We defer a formal definition
of the \GKS to \cref{sec:gks-preliminaries}. Later, in
\cref{sec:phase-based} and \cref{sec:spaces}, we briefly
sketch the structural claims concerning the \GKS given by Bansal et
al.~\cite{BaElKN18}. Using this structural information, in
\cref{sec:reduction}, we link the \GKS to the Hydra game: we show
that a (randomized) algorithm of total cost~$R$ for the Hydra game on a
specific tree (called \emph{factorial tree}) implies a~(randomized)
$(R+1)$-competitive solution for the \GKS. This, along with the performance
guarantees of \ouralg given in \cref{sec:hydra}, yields the desired
competitiveness bound. We remark that while the explicit definition of the
Hydra game is new, the algorithm of Bansal et al.~\cite{BaElKN18} easily
extends to its framework.

Finally, in \cref{sec:lower}, we give an explicit lower bound
construction for the \GKS, which does not use a reduction to the metrical task system
problem, hereby improving the bound from $\Omega(k / \log k)$ to $\Omega(k)$.

\section{Hydra Game}
\label{sec:hydra}

The Hydra game\footnote{This is a work of science. Any resemblance of the
process to the decapitation of a mythical many-headed serpent-shaped monster is
purely coincidental.} is played between an online algorithm and an adversary on
a fixed unweighted tree \tree, known to the algorithm in advance. The nodes of
\tree have states which change throughout the game: Each node can be either
\emph{asleep}, \emph{alive} or \emph{dead}. Initially, the root~\treeroot is
alive and all other nodes are asleep.  At all times, the following invariant is
preserved: all ancestors of alive nodes are dead and all their descendants are
asleep. In a~single step, the adversary picks a single alive node $w$, kills it
(changes its state to dead) and makes all its (asleep) children alive. Note that
such adversarial move preserves the invariant above.

An algorithm must remain at some alive node (initially, it is at the root
\treeroot). If~an~algorithm is at a node $w$ that has just been killed, it has
to move to any still alive node $w'$ of its choice. For such movement it pays
$\dist(w,w')$, the length of the shortest path between $w$ and $w'$ in the
tree $T$. The game ends when all nodes except one (due to the invariant, it has to
be an~alive leaf) are dead. Unlike many online problems, here our sole goal is to minimize 
the total (movement) cost of an online algorithm (i.e., without comparing it
to the cost of the offline optimum). 

This game is not particularly interesting in the deterministic setting: As an
adversary can always kill the node where a deterministic algorithm resides,
the algorithm has to visit all but one nodes of tree $T$, thus paying
$\Omega(|\tree|)$. On the other hand, a trivial DFS traversal of tree
\tree has the cost of $O(|\tree|)$. Therefore, we focus on randomized
algorithms and assume that the adversary is oblivious: it knows an online
algorithm, but not the random choices made by it thus far.

%%%%%%%%%%%%%%%%%%%%%%%%%%%%%%%%%%%%%%%%%%%%%%%%%%%%%%%%%%%%%%%%%%%%%%%%%%%%%%%

\subsection{Randomized Algorithm Definition}
\label{sec:ouralg}
It is convenient to describe our randomized algorithm \ouralg as maintaining
a~probability distribution $\prob$ over set of nodes, where for any node $u$,
$\prob(u)$ denotes the probability that \ouralg is at $u$. We require that
$\prob(u) = 0$ for any non-alive node $u$. Whenever \ouralg decreases the
probability at a given node $u$ by $p$ and increases it at another node $w$ by
the same amount, we charge cost $p \cdot \dist(u,w)$ to \ouralg. By a
straightforward argument, one can convert such description into a more
standard, ``behavioral'' one, which describes randomized actions conditioned
on the current state of an algorithm, and show that the costs of both
descriptions coincide. We present the argument in \cref{app:app} for
completeness.

At any time during the game, for any node $u$ from tree $T$, $\rank(u)$ denotes
the number of non-dead (i.e., alive or asleep) leaves in the subtree rooted at
$u$.  As \ouralg knows tree \tree in advance, it knows node ranks as well.
Algorithm \ouralg maintains $\prob$ that is distributed over all alive nodes
proportionally to their ranks.  As all ancestors of an alive node are dead and
all its descendants are asleep, we have $\prob(u) = \rank(u)/\rank(\treeroot)$
if $u$ is alive and $\prob(u) = 0$ otherwise.  In particular, at the beginning
$\prob$ is $1$ at the root and $0$ everywhere else.

While this already defines the algorithm, we still discuss its behavior when an
alive node $u$ is killed by the adversary.  By \ouralg definition, we can think
of the new probability distribution $\prob'$ as obtained from $\prob$ in the
following way. First, \ouralg sets $\prob'(u) = 0$. Next, the probability
distribution at other nodes is modified as follows.
\begin{description}
\item[Case 1.] Node $u$ is not a leaf. \ouralg
    distributes the probability of $u$ among all (now alive) children of~$u$
    proportionally to their ranks, i.e., sets 
    $\prob'(w) = (\rank(w) / \rank(u)) \cdot \prob(u)$ for each child $w$ of $u$.
    
\item[Case 2.] Node $u$ is a leaf.  Note that there were some other non-dead
    leaves, as otherwise the game would have ended before this step, and
    therefore $\prob(u) \leq 1/2$.
    \ouralg distributes~$\prob(u)$ among all other nodes,
    scaling the probabilities of the remaining nodes up by a~factor of
    $1/(1-\prob(u))$. That is, it sets $\prob'(w) = \prob(w)/(1-\prob(u))$ for
    any node $w$.
\end{description}

Note that in either case, $\prob'$ is a valid probability distribution, i.e.,
all probabilities are non-negative and sum to $1$. Moreover, $\prob'$ is
distributed over alive nodes proportionally to their new ranks, and is equal
to zero at non-alive nodes.

\begin{observation}
\label{obs:leaf_prob}
At any time, the probability of an alive leaf $u$ is exactly $\prob(u) = 1/ \rank(\treeroot)$.
\end{observation}

%%%%%%%%%%%%%%%%%%%%%%%%%%%%%%%%%%%%%%%%%%%%%%%%%%%%%%%%%%%%%%%%%%%%%%%%%%%%%%%

\subsection{Analysis}

For the analysis, we need a few more definitions. We denote the height and the
number of the leaves of tree \tree by $\height$ and $\leaves$, respectively.
Let $\level(u)$ denote the height of the subtree rooted at $u$, where leaves
are at level $0$. Note that $\height = \level(\treeroot)$. 

To bound the cost of \ouralg, we define a potential $\Phi$, which is a
function of the current state of all nodes of $T$ and the current probability
distribution $\prob$ of \ouralg. We show that $\Phi$~is initially $O(\height
\cdot (1+\log \leaves))$, is always non-negative, and the cost of each
\ouralg's action can be covered by the decrease of $\Phi$. This will show that
the total cost of \ouralg is at most the initial value of $\Phi$, i.e.,
$O(\height \cdot (1+\log \leaves))$.

Recall that $\prob(w) = 0$ for any non-alive node $w$ and that $\rank(u)$ is
the number of non-dead leaves in the subtree rooted at $u$.  Specifically,
$\rank(\treeroot)$ is the total number of non-dead leaves in $T$. The potential is
defined as
\begin{equation}
   \Phi = 4 \cdot \height \cdot H(\rank(\treeroot)) + \sum_{w\in T} \prob(w) \cdot \level(w)\,,
\end{equation}
where $H(n) = \sum_{i=1}^n 1/i$ is the $n$-th harmonic number.

\begin{lemma}\label{lem:PotUB}
At any time, $\Phi = O(\height \cdot (1 + \log \leaves))$.
\end{lemma}

\begin{proof}
Since $\rank(\treeroot) \leq \leaves$ at all times, the first summand of
$\Phi$ is $O(\height \cdot \log \leaves)$. The second summand of $\Phi$ is a
convex combination of node levels, which range from $0$ to $\height$, and
is thus bounded by $\height$.
\end{proof}

\begin{lemma}
\label{lem:algvsPot}
Fix any step in which an adversary kills a node $u$ and in result \ouralg
changes the probability distribution from $\prob$ to $\prob'$. Let $\Delta \ouralg$ be 
the cost incurred in this step by \ouralg and let $\Delta \Phi$ be 
the resulting change in the potential $\Phi$. Then, $\Delta \Phi \leq - \Delta \ouralg$.
\end{lemma}

\begin{proof}
We denote the ranks before and after the adversarial event by $\rank$ and $\rank'$, respectively.
We consider two cases depending on the type of $u$.

\begin{description}
\item[Case 1.]
The killed node $u$ is an internal node.
In this case, $\Delta \ouralg = \prob(u)$ as \ouralg simply moves the total probability 
of $\prob(u)$ along a distance of one (from $u$ to its children).
As $\rank'(\treeroot) = \rank(\treeroot)$, the first summand of $\Phi$ remains unchanged.
Let $C(u)$ be the set of children of $u$. Then,
\begin{align*}
\Delta \Phi 
    & = \sum_{w \in T} (\prob'(w) - \prob(w)) \cdot \level(w) = - \prob(u) \cdot \level(u) 
        + \sum_{w \in C(u)} \prob'(w) \cdot \level(w) \\
    & \leq - \prob(u) \cdot \level(u) + \sum_{w \in C(u)} \prob'(w) \cdot (\level(u)-1) \\
    & = - \prob(u) \cdot \level(u) + \prob(u) \cdot (\level(u) - 1) = - \Delta \ouralg \,,
\end{align*}
where the inequality holds as level of a node is smaller than the level of its parent and the penultimate equality follows as the whole probability mass at $u$ 
is distributed to its children.

\item[Case 2.]
The killed node $u$ is a leaf. It is not the last alive node, as in such case
the game would have ended before, i.e., it holds that $\rank(\treeroot) \geq
2$. \ouralg moves the probability of $\prob(u) =
1/\rank(\treeroot)$ (cf.~\cref{obs:leaf_prob}) along a distance of
at most $2 \cdot \height$, and thus $\Delta \ouralg \leq 2 \cdot \height /
\rank(\treeroot)$.

Furthermore, for any $w \neq u$, $\prob'(w) = \prob(w)/(1-\prob(u))$. Using
$\prob(u) = 1/\rank(\treeroot)$, we infer that the probability at a node $w
\neq u$ increases by
\begin{align}
    \nonumber
    \prob'(w) - \prob(w) 
    = &\; \left(\frac{1}{1-\prob(u)} - 1 \right) \cdot \prob(w) 
    = \frac{\prob(u)}{1-\prob(u)} \cdot \prob(w) \\
    \label{eq:phi_prob_change}            
    = &\; \frac{1}{\rank(\treeroot)-1} \cdot \prob(w)
    \leq \frac{2}{\rank(\treeroot)} \cdot \prob(w) \,,
\end{align}
where the last inequality follows as $\rank(\treeroot) \geq 2$.

Using \eqref{eq:phi_prob_change} and the relation $\rank'(\treeroot) =
\rank(\treeroot) - 1$ (the number of non-dead leaves decreases by $1$), we
compute the change of the potential:
\begin{align*}
\Delta \Phi 
    & = 4 \cdot \height \cdot \left( H(\rank'(\treeroot)) - H(\rank(\treeroot)) \right)
        + \sum_{w \in T} (\prob'(w) - \prob(w)) \cdot \level(w) \\
    & = -\frac{4 \cdot \height}{\rank(\treeroot)} 
        + (\prob'(u) - \prob(u)) \cdot \level(u) 
        + \sum_{w \neq u} (\prob'(w) - \prob(w)) \cdot \level(w) \\
    & \leq -\frac{4 \cdot \height}{\rank(\treeroot)} 
        + \sum_{w \neq u} \frac{2}{\rank(\treeroot)} \cdot \prob(w) \cdot \height 
     \leq -\frac{2 \cdot \height}{\rank(\treeroot)} 
     \leq -\Delta \ouralg \,.
\end{align*}
In the first inequality, we used that $\level(u) = 0$ and $\level(w) \leq \height$ for any $w$.
\end{description}
Summing up, we showed that $\Delta \Phi \leq - \Delta \ouralg$ in both cases.
\end{proof}

\begin{theorem}\label{thm:algcost}
For the Hydra game played on any tree \tree of height $\height$ and $\leaves$
leaves, the total cost of \ouralg is at most $O(\height \cdot (1+\log
\leaves))$.
\end{theorem}

\begin{proof}
Let $\Phi_\mathrm{B}$ denote the initial value of $\Phi$. By non-negativity of $\Phi$
and \cref{lem:algvsPot}, it holds that the total cost of \ouralg is at
most $\Phi_\mathrm{B}$. The latter amount is at most $O(\height \cdot (1+\log
\leaves))$ by \cref{lem:PotUB}.
\end{proof}

Although \ouralg and \cref{thm:algcost} may seem simple, when applied
to appropriate trees, they yield improved bounds for the \GKS in uniform
metrics, as shown in the next section.

\section{Improved Algorithm for Generalized k-Server Problem}
\label{sec:k-server}

In this part, we show how any solution for the Hydra game on a specific tree
(defined later) implies a solution to the \GKS in uniform metrics. This will
yield an $O(k^2 \log k)$-competitive randomized algorithm for the \GKS,
improving the previous bound of $O(k^3 \cdot \log k)$~\cite{BaElKN18}. We note
that this reduction is implicit in the paper of Bansal et al.~\cite{BaElKN18}, so our
contribution is in formalizing the Hydra game and solving it more efficiently.

%%%%%%%%%%%%%%%%%%%%%%%%%%%%%%%%%%%%%%%%%%%%%%%%%%%%%%%%%%%%%%%%%%%%%%

\subsection{Preliminaries}
\label{sec:gks-preliminaries}

The generalized $k$-server problem in uniform metrics is formally defined as
follows. The offline part of the input comprises $k$ uniform metric spaces
$M_1,\dots,M_k$. The metric $M_i$ has $n_i \geq 2$ points, the distance
between each pair of its points is $1$. There are $k$ servers denoted
$s_1,\dots,s_k$, the server $s_i$ starts at some fixed point in $M_i$
and always remains at some point of $M_i$.

The online part of the input is a sequence of requests, each request being a
$k$-tuple $(r_1,\dots,r_k) \in \prod_{i=1}^k M_i$. To service a request, an
algorithm needs to move its servers, so that at least one server $s_i$ ends at
the request position $r_i$. Only after the current request is serviced, an
online algorithm is given the next one.

The cost of an algorithm \ALG on input $I$, denoted $\ALG(I)$, is the total
distance traveled by all its $k$ servers.  We say that a randomized online
algorithm \ALG is $\beta$-competitive if there exists a~constant $\gamma$, such
that for any input $I$, it holds that $\E[\ALG(I)] \leq \beta \cdot \OPT(I) +
\gamma$, where the expected value is taken over all random choices of \ALG, and
where $\OPT(I)$ denotes the cost of an~optimal \emph{offline} solution for input
$I$. The constant $\gamma$ may be a function of $k$, but it cannot depend on an
online part of the input.

%%%%%%%%%%%%%%%%%%%%%%%%%%%%%%%%%%%%%%%%%%%%%%%%%%%%%%%%%%%%%%%%%%%%%%%%%%%%%%%

\subsection{Phase-Based Approach}
\label{sec:phase-based}

We start by showing how to split the sequence of requests into phases. To this
end, we need a few more definitions. A (server) \emph{configuration} is a
$k$-tuple $c = (c_1,\dots,c_k) \in \prod_{i=1}^k M_i$, denoting positions of
respective servers. For a request $r = (r_1,\dots,r_k) \in \prod_{i=1}^k M_i$,
we define the set of \emph{compatible} configurations $\comp(r) =
\{ (c_1,\dots,c_k) : \exists_i\,c_i = r_i \}$, i.e., the set of all
configurations that can service the request $r$ without moving a~server. Other
configurations we call \emph{incompatible} with $r$.

An input is split into phases, with the first phase starting with the
beginning of an~input. The phase division process described below is
constructed to ensure that \OPT pays at least~$1$ in any phase, perhaps except
the last one. At the beginning of a phase, all configurations are
\emph{phase-feasible}. Within a~phase, upon a request $r$, all configurations
incompatible with~$r$ become \emph{phase-infeasible}. The phase ends once all
configurations are phase-infeasible; if this is not the end of the input, the
next phase starts immediately, i.e., all configurations are restored to the
phase-feasible state before the next request.  Note that the description above
is merely a way of splitting an~input into phases and marking configurations
as phase-feasible and phase-infeasible. The actual description of an online
algorithm will be given later.

Fix any finished phase and any configuration $c$ and consider an algorithm
that starts the phase with its servers at configuration $c$. When
configuration $c$ becomes phase-infeasible, such algorithm is forced to move
and pay at least $1$. As each configuration eventually becomes
phase-infeasible in a finished phase, any algorithm (even \OPT) must pay at
least $1$ in any finished phase. Hence, if the cost of a phase-based algorithm
for servicing requests of a single phase can be bounded by $f(k)$, the
competitive ratio of this algorithm is then at most $f(k)$.

%%%%%%%%%%%%%%%%%%%%%%%%%%%%%%%%%%%%%%%%%%%%%%%%%%%%%%%%%%%%%%%%%%%%%%%%%%%%%%%

\subsection{Configuration Spaces}
\label{sec:spaces}

Phase-based algorithms that we construct will not only track the set of
phase-feasible configurations, but they will also group these configurations in
certain sets, called \emph{configuration spaces}.

To this end, we introduce a special \emph{wildcard} character \wc.
Following~\cite{BaElKN18}, for any $k$-tuple $q = (q_1,\dots,q_k) \in
\prod_{i=1}^k (M_i \cup \{ \wc \})$, we define a \emph{(configuration) space}
$\cspace{q} = \{ (c_1,\dots,c_k) \in \prod_{i=1}^k M_i : \forall_i \; c_i = q_i
\vee q_i = \wc \}$. A coordinate with $q_i = \wc$ is called \emph{free} for the
configuration space $\cspace{q}$. That is, $\cspace{q}$ contains all
configurations that agree with $q$ on all non-free coordinates. 

The number of free coordinates in $q$ defines the \emph{dimension} of
$\cspace{q}$ denoted $\dim(\cspace{q})$. Observe that the $k$-dimensional space
$\cspace{(\wc,\dots,\wc)}$ contains all configurations. If tuple~$q$ has no \wc
at any position, then $\cspace{q}$ is $0$-dimensional and contains only
(configuration) $q$. The following lemma, proven by Bansal et
al.~\cite{BaElKN18}, follows immediately from the definition of configuration
spaces.

\begin{lemma}[Lemma 3.1 of \cite{BaElKN18}]
\label{lem:subspaces}
Let $\cspace{q}$ be a $d$-dimensional configuration space (for some $d \geq
0$) whose all configurations are phase-feasible. Fix a request $r$. If there
exists a~configuration in $\cspace{q}$ that is not compatible with $r$, then
there exist $d$ (not necessarily disjoint) subspaces
$\cspace{q_1},\dots,\cspace{q_d}$, each of dimension $d - 1$, such that
$\bigcup_i \cspace{q_i} = \cspace{q} \cap \comp(r)$. Furthermore, for all $i$,
the $k$-tuples $q_i$ and $q$ differ exactly at one position.
\end{lemma}

Using the lemma above, we may describe a way for an online algorithm to keep
track of all phase-feasible configurations. To this end, it maintains a set
$\SET$ of (not necessarily disjoint) configuration spaces, such that their
union is exactly the set of all phase-feasible configurations. We call spaces
from $\SET$ \emph{alive}.

At the beginning, $\SET = \{ \cspace{(\wc,\dots,\wc)} \}$. Assume now that
a~request $r$ makes some configurations from a $d$-dimensional space
$\cspace{q} \in \SET$ phase-infeasible. (A request may affect many spaces from
$\SET$; we apply the described operations to each of them sequentially in an
arbitrary order.) In such case, $\cspace{q}$ stops to be alive, it is removed
from $\SET$ and till the end of the phase it will be called \emph{dead}. Next,
we apply \cref{lem:subspaces} to $\cspace{q}$, obtaining 
$d$~configuration spaces $\cspace{q_1}, \dots, \cspace{q_d}$, such that their
union is $\cspace{q} \cap \comp(r)$, i.e., contains all those configurations
from $\cspace{q}$ that remain phase-feasible. We make all spaces
$\cspace{q_1}, \dots, \cspace{q_d}$ alive and we insert them into $\SET$.
(Note that when $d = 0$, set $\cspace{q}$ is removed from \SET, but no space
is added to it.) This way we ensure that the union of spaces from $\SET$ remains
equal to the set of all phase-feasible configurations. Note that when a phase
ends, $\SET$ becomes empty. We emphasize that the evolution of set $\SET$
within a phase depends only on the sequence of requests and not on the
particular behavior of an online algorithm.

%%%%%%%%%%%%%%%%%%%%%%%%%%%%%%%%%%%%%%%%%%%%%%%%%%%%%%%%%%%%%%%%%%%%%%%%%%%%%%%

\subsection{Factorial Trees: From Hydra Game to Generalized k-Server}
\label{sec:reduction}

Given the framework above, an online algorithm may keep track of the set of
alive spaces $\SET$, and at all times try to be in a configuration from some
alive space. If this space becomes dead, an algorithm changes its configuration
to any configuration from some other alive space from $\SET$. 

The crux is to choose an appropriate next alive space. To this end, our
algorithm for the \GKS will internally run an instance of the Hydra game (a new
instance for each phase) on a special tree, and maintain a mapping from alive
and dead spaces to alive and dead nodes in the tree. Moreover, spaces that are
created during the algorithm runtime, as described in
\cref{sec:spaces}, have to be dynamically mapped to tree nodes that
were so far asleep.

In our reduction, we use a \emph{$k$-factorial} tree. It has height $k$ (the
root is on level $k$ and leaves on level $0$). Any node on level $d$ has exactly
$d$ children, i.e., the subtree rooted at a $d$-level node has $d!$ leaves,
hence the tree name. On the $k$-factorial tree, the total cost of \ouralg is
$O(k \cdot (1+\log k!)) = O(k^2 \cdot \log k)$. We now show that this implies an
improved algorithm for the \GKS.

\begin{theorem}
\label{thm:reduction}
If there exists a (randomized) online algorithm \alghyd for the Hydra game on
the $k$-factorial tree of total (expected) cost $R$, then there exists a
(randomized) $(R+1)$-competitive online algorithm \alggks for the \GKS in
uniform metrics.
\end{theorem}

\begin{proof}
Let $I$ be an input for the generalized $k$-server problem in uniform metric
spaces. \alggks~splits~$I$ into phases as described in
\cref{sec:phase-based} and, in each phase, it tracks the phase-feasible
nodes using set $\SET$ of alive spaces as described in
\cref{sec:spaces}. For each phase, \alggks~runs a new instance $I_H$ of
the Hydra game on a $k$-factorial tree $T$, translates requests from $I$ to
adversarial actions in $I_H$, and reads the answers of \alghyd executed on
$I_H$. At all times, \alggks~maintains a~(bijective) mapping from alive
(respectively, dead) $d$-dimensional configuration spaces to alive
(respectively, dead) nodes on the $d$-th level of the tree $T$. In particular,
at the beginning, the only alive space is the $k$-dimensional space
$\cspace{(\wc,\dots,\wc)}$, which corresponds to the tree root (on level $k$).
The configuration of \alggks will always be an~element of the space
corresponding to the tree node containing~\alghyd. More precisely, within each
phase, a request $r$ is processed in the following way by \alggks.
\begin{itemize}
\item Suppose that request $r$ does not make any configuration
phase-infeasible. In this case, \alggks~services $r$ from its current
configuration and no changes are made to \SET. Also no adversarial actions are
executed in the Hydra game.

\item Suppose that request $r$ makes some (but not all) configurations
phase-infeasible. We assume that this kills only one $d$-dimensional
configuration space~$\cspace{q}$. (If $r$ causes multiple configuration spaces
to become dead, \alggks processes each such killing event separately, in
an~arbitrary order.)

By the description given in \cref{sec:spaces}, $\cspace{q}$ is then
removed from \SET and $d$~new $(d-1)$-dimensional spaces $\cspace{q_1}, \dots,
\cspace{q_d}$ are added to~\SET. \alggks executes appropriate adversarial
actions in the Hydra game: a node $v$ corresponding to $\cspace{q}$ is killed
and its $d$ children on level $d-1$ change state from asleep to alive. \alggks
modifies the mapping to track the change of $\SET$: (new and now alive) spaces
$\cspace{q_1}, \dots, \cspace{q_d}$ become mapped to (formerly asleep and now
alive) $d$ children of $v$. Afterwards, \alggks observes the answer of
algorithm~\alghyd on the factorial tree and replays it. Suppose \alghyd moves
from (now dead) node $v$ to an alive node~$v'$, whose corresponding space is
$\cspace{q'} \in \SET$. In this case, \alggks changes its configuration to the
closest configuration (requiring minimal number of server moves) from
$\cspace{q'}$. It remains to relate its cost to the cost of \alghyd. By
\cref{lem:subspaces} (applied to spaces corresponding to all nodes on the
tree path from $v$ to $v'$), the corresponding $k$-tuples $q$, $q'$ differ on at
most $\dist(v,v')$ positions. Therefore, adjusting the configuration of \alggks,
so that it becomes an~element of~$\cspace{q'}$, requires at most $\dist(v,v')$
server moves, which is exactly the cost~of~\alghyd.

Finally, note that when \alggks processes all killing events, it ends in a
configuration of an~alive space, and hence it can service the
request $r$ from its new configuration.

\item Suppose that request $r$ makes all remaining configurations
phase-infeasible. In such case, \alggks moves an arbitrary server to service
this request, which incurs a cost of $1$. In this case, the current phase ends,
a new one begins, and \alggks initializes a new instance of the Hydra game. 
\end{itemize}

Let $f \geq 1$ be the number of all phases for input $I$ (the last one may be
not finished). The cost of \OPT in a single finished phase is at least $1$. By
the reasoning above, the (expected) cost of \alggks in a single phase is at most
$R+1$. Therefore, $\E[\alggks(I)] \leq (R+1) \cdot f \leq (R+1) \cdot \OPT(I) +
(R+1)$, which completes the proof.
\end{proof}

Using our algorithm \ouralg for the Hydra game along with the reduction given
by \cref{thm:reduction} immediately implies the following result.

\begin{corollary}
There exists a randomized $O(k^2 \cdot \log k)$-competitive online algorithm for
the \GKS in uniform metrics.
\end{corollary}

\section{Lower bound}
\label{sec:lower}

Next, we show that that competitive ratio of any (even randomized) 
online algorithm for the \GKS in uniform metrics is at least $\Omega(k)$, 
as long as each metric space~$M_i$ contains at least
two points.  For each $M_i$, we choose two distinct points, the
initial position of the $i$-th server, which we denote $0$ and any other
point, which we denote $1$. The adversary is going to issue only requests
satisfying $r_i \in \{0,1\}$ for all~$i$, hence without loss of generality any
algorithm will restrict its server's position in each $M_i$ to $0$ and $1$.
(To see this, assume without loss of generality that the algorithm is lazy,
i.e., it is only allowed to move when a request is not covered by any of its
server, and is then allowed only to move a single server to cover that request.) 
For this reason, from now on we assume that $M_i = \{0,1\}$ for all $i$,
ignoring superfluous points of the metrics. 

The configuration of any algorithm can be then encoded using a binary word of
length~$k$. It is convenient to view all these $2^k$ words (configurations) as
nodes of the $k$-dimensional hypercube: two words are connected by a hypercube
edge if they differ at exactly one position. Observe that a cost of changing
configuration $c$ to $c'$, denoted $\dist(c,c')$ is exactly the distance
between $c$ and~$c'$ in the hypercube, equal to the number of positions on
which the corresponding binary strings differ. 

In our construction, we compare the cost of an online algorithm to the cost of
an~algorithm provided by the adversary. Since $\OPT$'s cost can be only lower
than the latter, such approach yields a lower bound on the performance of
the online algorithm.

For each word $w$, there is exactly one word at distance $k$, which we call
its \emph{antipode} and denote $\bar{w}$. Clearly, $\bar{w}_i = 1-w_i$ for
all~$i$. Whenever we say that an adversary \emph{penalizes} configuration $c$,
it issues a request at $\bar{c}$. An algorithm that has servers at
configuration~$c$ needs to move at least one of them. On the other hand, any
algorithm with servers at configuration $c' \neq c$ need not move its servers;
this property will be heavily used by an~adversary's algorithm.

%%%%%%%%%%%%%%%%%%%%%%%%%%%%%%%%%%%%%%%%%%%%%%%%%%%%%%%%%%%%%%%%%%%%%%%%%%%%%%%

\subsection{A Warm-Up: Deterministic Algorithms}

To illustrate our general framework, we start with a description of an
$\Omega(2^k / k)$ lower bound that holds for any deterministic algorithm
\DET~\cite{BaElKN18}. (A more refined analysis yields a better lower bound of
$2^k - 1$~\cite{KouTay04}.) The adversarial strategy consists of a sequence of
independent identical phases. Whenever \DET is in some configuration, the
adversary penalizes this configuration. The phase ends when $2^k-1$
\emph{different} configurations have been penalized. This means that \DET was
forced to move at least $2^k-1$ times, at a total cost of at least $2^k - 1$. In
the same phase, the adversary's algorithm makes only a single move (of cost at
most $k$) at the very beginning of the phase: it moves to the only configuration
that is not going to be penalized in the current phase. This shows that the
\DET-to-\OPT ratio in each phase is at least $(2^k - 1) / k$.

%%%%%%%%%%%%%%%%%%%%%%%%%%%%%%%%%%%%%%%%%%%%%%%%%%%%%%%%%%%%%%%%%%%%%%%%%%%%%%%

\subsection{Extension to Randomized Algorithms}
\label{sec:lower_rand}

Adopting the idea above to a randomized algorithm \RAND is not straightforward.
Again, we focus on a single phase and the adversary wants to leave (at least)
one configuration non-penalized in this phase. However, now the adversary only
knows \RAND's probability distribution $\rprob$ over configurations and not its
actual configuration. (At any time, for any configuration $c$, $\rprob(c)$ is
the probability that \RAND's configuration is equal to $c$.) We focus on
a~greedy adversarial strategy that always penalizes the configuration with
maximum probability. However, arguing that \RAND incurs a significant cost is
not as easy as for \DET.

First, the support of $\rprob$ can also include configurations that have been
already penalized by the adversary in the current phase. This is but a nuisance,
easily overcome by penalizing such configurations repeatedly if \RAND keeps
using them, until their probability becomes negligible. Therefore, in this
informal discussion, we assume that once a~configuration $c$ is penalized in a
given phase, $\rprob(c)$ remains equal to zero.

Second, a~straightforward analysis of the greedy adversarial strategy fails to
give a~non-trivial lower bound. Assume that $i \in \{0, \dots, 2^k-2\}$
configurations have already been penalized in a given phase, and the support of
$\rprob$ contains the remaining $2^k - i$ configurations. The maximum
probability assigned to one of these configurations is at least $1/(2^k-i)$.
When such configuration is penalized, \RAND needs to move at least one server
with probability at least $1/(2^k-i)$. With such bounds, we would then prove
that the algorithm's expected cost is at least $\sum_{i=0}^{2^k-2} 1/(2^k-i) =
\Omega(\log 2^k) = \Omega(k)$.  Since we bounded the adversary's cost per phase
by $k$, this gives only a constant lower bound.

What we failed to account is that the actual distance traveled by \RAND in a
single step is either larger than $1$ or \RAND would not be able to maintain a
uniform distribution over non-penalized configurations. However, actually
exploiting this property seems quite complex, and therefore we modify the
adversarial strategy instead.

The crux of our actual construction is choosing a subset $Q$ of the
configurations, such that $Q$ is sufficiently large (we still have $\log(|Q|) =
\Omega(k)$), but the minimum distance between any two points of $Q$ is
$\Omega(k)$. Initially, the adversary forces the support of $\rprob$ to be
contained in $Q$. Afterwards, the adversarial strategy is almost as described
above, but reduced to set $Q$ only. This way, in each step the support of
$\rprob$ is a set $S \subseteq Q$, and the adversary forces \RAND to move with
probability at least $1/|S|$ \emph{over a distance at least $\Omega(k)$}, which
is the extra $\Theta(k)$ factor. We begin by proving the existence of such a set
$Q$ for sufficiently large $k$. The proof is standard (see, e.g., Chapter 17 of
\cite{Jukna11}); we give it below for completeness.

\begin{lemma}
\label{lem:code}
For any $k \geq 16$, there exists a set $\LB \subseteq \{0,1\}^k$
of binary words of length $k$, satisfying the following two 
properties:
\begin{description}
	\item[size property:]{$|\LB| \geq 2^{k/2}/k$,}
	\item[distance property:]{$\dist(v,w) \geq k/16$ for any $v,w \in \LB$.} 
\end{description}
\end{lemma}

\begin{proof}
Let $\ell = \lfloor k/16 \rfloor \geq k / 32$. For any word $q$, we define its
$\ell$-neighborhood $B_\ell(q) = \{ w : \dist(q,w) \leq \ell \}$.

We construct set \LB greedily. We maintain set $\LB$ and set
$\Gamma(\LB) = \bigcup_{q \in Q} B_\ell(q)$. We start with $\LB = \emptyset$ (and
thus with $\Gamma(\LB) = \emptyset$). In each step, we extend \LB with an arbitrary
word $w \in \{0,1\}^k \setminus \Gamma(\LB)$ and update $\Gamma(\LB)$ accordingly. 
We proceed until set $\Gamma(\LB)$
contains all possible length-$k$ words. Clearly, the resulting set \LB satisfies the distance
property. 

It remains to show that $|\LB| \geq 2^{k/2}/k$. For a
word $q$, the size of $B_\ell(q)$ is
\begin{align*}
|B_\ell(q)| 
	= & \sum_{i = 0}^{\lfloor k/16 \rfloor} \binom{k}{i} 
	< k \cdot \binom{k}{\lfloor k/16 \rfloor} 
	\leq k \cdot \left (\frac{k\cdot \e}{\lfloor k/16 \rfloor} \right)^{\lfloor k/16 \rfloor} \\
	\leq &\; k \cdot \left (\frac{k\cdot \e}{k/32} \right)^{k/16} 
	= k \cdot \left((32 \cdot \e)^{1/8}\right)^{k/2} < k \cdot 2^{k/2} \,.
\end{align*}
That is, in a single step, $\Gamma(\LB)$ increases by at most $k \cdot
2^{k/2}$ elements. Therefore, the process continues for at least $2^k / (k
\cdot 2^{k/2}) = 2^{k/2} / k$ steps, and thus the size of \LB is at least
$2^{k/2} / k$.
\end{proof}

\begin{theorem}\label{thm:lb}
The competitive ratio of every (randomized) online algorithm solving the \GKS
in uniform metrics is at least $\Omega(k)$.
\end{theorem}

\begin{proof}
In the following we assume that $k \geq 16$, otherwise the theorem follows trivially. 
We fix any randomized online algorithm \RAND. The lower bound strategy
consists of a sequence of independent phases. Requests of each phase can be
(optimally) serviced with cost at most~$k$ and we show that \RAND's expected
cost for a single phase is $\Omega(k^2)$, i.e., the ratio between these costs
is $\Omega(k)$. As the adversary may present an arbitrary number of phases to
the algorithm, this shows that the competitive ratio of \RAND is $\Omega(k)$,
i.e., by making the cost of \RAND arbitrarily high, the additive constant in
the definition of the competitive ratio (cf.~\cref{sec:gks-preliminaries})
becomes negligible.

As in our informal introduction, $\rprob(c)$ denotes the probability that
\RAND has its servers in configuration $c$ (at time specified in the context).
We extend the notion $\rprob$ to sets, i.e., $\rprob(X) = \sum_{c \in X} \rprob(x)$
where $X$ is a set of configurations. We denote the complement of $X$ (to
$\prod_{i=1}^k M_i$) by $X^C$. We use $\eps = 2^{-(2 k + 2)}$ throughout the
proof.

To make the description concise, we define an auxiliary routine $\con(X)$ for
the adversary (for some configuration set $X$). In this routine, the adversary
repeatedly checks whether there exists a configuration $c \not \in X$, such that
$\rprob(x) > \eps$. In such case, it penalizes $c$; if no such configuration
exists, the routine terminates. We may assume that the procedure always
terminates after finite number of steps, as otherwise \RAND's competitive ratio
would be unbounded. (\RAND pays at least $\eps$ in each step of the routine
while an adversary's algorithm may move its servers to any configuration from
set $X$, and from that time service all requests of $\con(X)$ with no cost.) 

The adversarial strategy for a single phase is as follows. First, it constructs
$Q_1$ as the configuration set fulfilling the properties of
\cref{lem:code}; let $m$ denote its cardinality. The phase consists then
of $m$ executions of $\con$ routine: $\con(Q_1),\con(Q_2),\dots,\con(Q_m)$. For
$i \in \{2,\dots,m\}$, set $Q_i$ is defined in the following way. The adversary
observes \RAND's distribution $\rprob$ right after routine $\con(Q_{i-1})$
terminates; at this point this distribution is denoted $\rprob_{i-1}$. Then, the
adversary picks configuration $c_{i-1}$ to be the element of $Q_{i-1}$ that
maximizes the probability $\rprob_{i-1}$, and sets $Q_i = Q_{i-1} \setminus \{
c_{i-1} \}$.

We begin by describing the way that the adversary services the requests. Observe
that set $Q_m$ contains a single configuration, henceforth denoted $c^*$. The
configuration $c^*$ is contained in all sets $Q_1, \ldots, Q_m$, and thus $c^*$
is never penalized in the current phase. Hence, by moving to $c^*$ at the
beginning of the phase, which costs at most $k$, and remaining there till the
phase ends, the adversary's algorithm services all phase requests at no further
cost.

It remains to lower-bound the cost of \RAND. $\con(Q_1)$ may incur no cost; its
sole goal is to confine the support of $\rprob$ to $Q_1$. Now, we fix any $i \in
\{2,\dots,m\}$ and estimate the cost incurred by $\con(Q_i)$.  Recall that the
probability distribution right before $\con(Q_i)$ starts (and right after
$\con(Q_{i-1})$ terminates) is denoted $\rprob_{i-1}$ and the distribution right
after $\con(Q_i)$ terminates is denoted $\rprob_i$.

During $\con(Q_i)$ a probability mass $\rprob_{i-1}(c_{i-1})$, is moved from
$c_{i-1}$ to nodes of set $Q_i$ (recall that $Q_i \uplus \{c_{i-1}\} =
Q_{i-1}$). Some negligible amounts (at most $\rprob_i(Q^C_i)$) of this
probability may however remain outside of $Q_i$ after $\con(Q_i)$ terminates.
That is, \RAND moves at least the probability mass of $\rprob_{i-1}(c_{i-1}) -
\rprob_i(Q^C_{i})$ from configuration $c_{i-1}$ to configurations from $Q_i$
(i.e., along a distance of at least $\dist(c_{i-1},Q_i)$), Therefore, its
expected cost due to $\con(Q_i)$ is at least $(\rprob_{i-1}(c_{i-1}) -
\rprob_i(Q^C_{i}) ) \cdot \dist(c_{i-1},Q_i)$.

First, using the properties of $\con(Q_{i-1})$ and the definition of $c_{i-1}$, we obtain 
\begin{equation}
\label{eq:prob1}
	\rprob_{i-1}(c_{i-1}) 
		\geq \frac{\rprob_{i-1}(Q_{i-1})}{|Q_{i-1}|} 
		= \frac{1 - \rprob_{i-1}(Q^C_{i-1})}{|Q_{i-1}|} 
		\geq \frac{1 - |Q^C_{i-1}| \cdot \eps}{|Q_{i-1}|} 
	> \frac{1 - 2^{-(k+2)}}{|Q_{i-1}|} \,. 
\end{equation}
Second, using the properties of $\con(Q_i)$ yields 
\begin{equation}
\label{eq:prob2}
	\rprob_i(Q^C_i) \leq |Q^C_i| \cdot \eps < 2^{-(k+2)} = \frac{2^{-2}}{2^k} 
		< \frac{2^{-2}}{|Q_{i-1}|} \,. 
\end{equation}
Using \eqref{eq:prob1} and \eqref{eq:prob2}, we bound the expected cost of $\RAND$ 
due to routine $\con(Q_i)$ as
\begin{align}
	\nonumber
	\E[\RAND(\con(Q_i))]\;
		\geq\; & \left(\rprob_{i-1}(c_{i-1}) - \rprob_i(Q^C_{i})\right) \cdot \dist(c_{i-1},Q_i) \\
	\nonumber	
		\geq\; & \left(\frac{1-2^{-(k+2)}}{|Q_i|} - \frac{2^{-2}}{|Q_i|} \right)  \cdot \frac{k}{16} 
		\geq \frac{1}{2 \cdot |Q_i|} \cdot \frac{k}{16} \\
	\label{eq:prob3}	
		=\; & k / (32 \cdot (m-i+1))
\end{align}
The second inequality above follows as all configurations from $\{c_{i-1}\}
\uplus Q_i$ are distinct elements of $Q_1$, and hence their mutual distance is
at least $k/16$ by the distance property of~$Q_1$
(cf.~\cref{lem:code}). By summing \eqref{eq:prob3} over $i \in
\{2,\dots,m\}$, we obtain that the total cost of \RAND in a single phase is
$
	\E[\RAND] \geq \sum_{i=2}^m \E[\RAND(\con(Q_i))] \geq \frac{k}{32} \cdot 
		\sum_{i=2}^m \frac{1}{m-i+1} = \Omega(k \cdot \log m) = \Omega(k^2)
$.
The last equality holds as $m \geq 2^{k/2} / k$ by the size property of $Q_1$.
(cf.~\cref{lem:code}).
\end{proof}

%%%%%%%%%%%%%%%%%%%%%%%%%%%%%%%%%%%%%%%%%%%%%%%%%%%%%%%%%%%%%%%%%%%%%%%%%%%%%%%

\section{Final remarks}

In this paper, we presented an abstract Hydra game whose solution we applied to
create an~algorithm for the \GKS. Any improvement of our \ouralg strategy for
the Hydra game would yield an improvement for the \GKS. However, we may show
that on a wide class of trees (that includes factorial trees used in our
reduction), \ouralg is optimal up to a constant factor. Thus, further improving
our upper bound of $O(k^2 \log k)$ for the \GKS will require another approach.

A lower bound for the cost of any randomized strategy for the Hydra game is
essentially the same as our single-phase construction from
\cref{sec:lower_rand} for the \GKS. That is, the adversary fixes a subset $Q$
of tree leaves, makes only nodes of $Q$ alive (this forces the algorithm to be
inside set $Q$), and then iteratively kills nodes of $Q$ where the algorithm is
most likely to be. As in the proof from \cref{sec:lower_rand}, such adversarial
strategy incurs the cost of $\Omega(\mindist(Q) \cdot \log |Q|)$, where
$\mindist(Q) = \min_{u \neq v\in Q} \dist(u,v)$.

The construction of appropriate $Q$ for a tree $T$ of depth $k = \height$ (be
either the $k$-factorial tree or the complete $k$-ary tree) is as follows. 
Let $Z$ be the set of all nodes of $T$ at level $\lfloor k/2
\rfloor$; for such trees, $\log |Z| = \Omega(\log \leaves)$. Let $Q$
consist of $|Z|$ leaves of the tree, one per node of~$Z$ chosen arbitrarily from
its subtree. Then, $\mindist(Q) = \Omega (\height)$ and $\log |Q| = \Omega(\log
\leaves)$, and thus the resulting lower bound $\Omega(\height \cdot \log \leaves)$ 
on the cost asymptotically matches the performance of \ouralg from \cref{thm:algcost}.

\bibliographystyle{plainurl}
\bibliography{references}

\appendix

\section{Probability Distribution and Algorithms}
\label{app:app}

When we described our algorithm \ouralg for the Hydra game, we assumed that its current
position in the tree is a~random node with probability distribution given by $\prob$.
In a single step, \ouralg decreases probability at some node $u$ from
$\prob(u)$ to zero and increases the probabilities of some other nodes $w_1,\ldots,
w_\ell$ by a total amount of $\prob(u)$. Such change can be split into 
$\ell$~elementary changes, each decreasing the probability at node $u$ by $p_i$ and
increasing it at node $w_i$ by the same amount. Each elementary change can be
executed and analyzed as shown in the following lemma.

\begin{lemma}
Let $\prob$ be a probability distribution describing the position of \ALG in the tree.
Fix two tree nodes, $u$ and $w$. Suppose $\prob'$ is a probability
distribution obtained from $\prob$ by decreasing $\prob(u)$ by $p$ and
increasing $\prob(w)$ by $p$. Then, \ALG can change its random position, so that it will be 
described by $\prob'$, and the expected cost of such change is $p\cdot \dist(u,w)$.
\end{lemma} 

\begin{proof}
We define \ALG's action as follows: if \ALG is at node $u$, then with
probability $p/\prob(u)$ it moves to node $w$. If \ALG is at some other node
it does not change its position.

We observe that the new distribution of \ALG is exactly $\prob'$. Indeed, the
probability of being at node $u$ decreases by $\prob(u) \cdot p /
\prob(u) = p$, while the probability of being at node $w$ increases by the
same amount. The probabilities for all nodes different than $u$ or $w$ remain
unchanged.

Furthermore, the probability that \ALG moves is $\prob(u) \cdot (p / \prob(u))
= p$ and the traveled distance is $\dist(u,w)$. The expected cost of the move
is then $p \cdot \dist(u,w)$, as desired.
\end{proof}

\end{document}